\documentclass[10pt,english]{IEEEtran}
\usepackage[T1]{fontenc}
\usepackage[latin9]{inputenc}
\usepackage[a4paper]{geometry}
\usepackage{amsmath}
\usepackage{graphicx}
\usepackage{amssymb}
\makeatletter
\newtheorem{example}{Example}
\newtheorem{definitn}{Definition}
\newtheorem{remrk}{Remark}
\newtheorem{thm}{Theorem}
\newtheorem{lemma}{Lemma}
\newtheorem{prop}{Proposition}

\def\ee{{\epsilon}}

\def\be{\begin{equation}}
\def\bea{\begin{eqnarray}}
\def\beas{\begin{eqnarray*}}
\def\eea{\end{eqnarray}}
\def\eeas{\end{eqnarray*}}

\def\bi{\begin{itemize}}
\def\ee{\end{equation}}
\def\ei{\end{itemize}}

\def\z1{z^{-1}}
\def\la{\label}

\def\mR{\mathbb{R}}

\def\diag{\mbox{ diag } }

\makeatother

\usepackage{babel}

\begin{document}

\title{On the Quality of Wireless Network Connectivity}
\author{\begin{tabular}{cc}
Soura Dasgupta & Guoqiang Mao\tabularnewline
Department of Electrical and Computer Engineering  & School of Electrical and Information Engineering\tabularnewline
The University of Iowa & The University of Sydney\tabularnewline
 & National ICT Australia\tabularnewline
\end{tabular}\emph{}%
\thanks{This research is partially supported by US NSF grants ECS-0622017,
CCF-072902, and CCF-0830747.}
\thanks{This research is partially supported by ARC Discovery project
DP110100538 and by the Air Force Research Laboratory, under agreement number FA2386-10-1-4102.
The U.S. Government is authorized to reproduce and distribute reprints for
Governmental purposes notwithstanding any copyright notation thereon. The
views and conclusions contained herein are those of the authors and should not
be interpreted as necessarily representing the official policies or endorsements,
either expressed or implied, of the Air Force Research Laboratory or the U.S.
Government.}}
\maketitle
\begin{abstract}
Despite intensive research in the area of network connectivity, there is an important category
of problems that remain unsolved: how to measure the
\emph{quality of connectivity} of a wireless multi-hop network which has
 a realistic number of nodes,  \emph{not necessarily}
large enough to warrant the use of asymptotic analysis, and has unreliable
connections, reflecting the inherent unreliable
characteristics of wireless communications? The quality of connectivity measures how easily and reliably a packet sent by a  node can reach another  node. It complements the use of \emph{capacity} to measure the quality of a network in saturated traffic scenarios and provides a native measure of the quality of (end-to-end) network connections. In this paper, we explore the use of probabilistic connectivity matrix as a possible tool to measure the quality of network connectivity.
Some interesting properties of the probabilistic connectivity matrix
and their connections to the quality of connectivity are demonstrated. We argue that the largest eigenvalue of the probabilistic connectivity matrix can serve as a good measure of the quality of network connectivity.
\end{abstract}
\begin{keywords}
Connectivity, network quality, probabilistic connectivity matrix
\end{keywords}
\thispagestyle{empty}
\section{Introduction\label{sec:Introduction}}

Connectivity is one of the most fundamental properties of wireless
multi-hop networks \cite{Gupta98Critical,Haenggi09Stochastic,Penrose03Random},
and is a prerequisite for providing many network functions, e.g. routing,
scheduling and localization. A network is said to be \emph{connected}
if and only if (iff) there is a (multi-hop) path between any pair of nodes. Further, a network is said to be $k$-connected iff there are $k$ mutually independent paths between any pair of nodes that do not share any node in common except the starting and the ending nodes. $k$-connectivity is often required for robust operations of the network.

There are two general approaches to studying the connectivity problem. 
The first, spearheaded by the seminal work of Penrose \cite{Penrose03Random}
and Gupta and Kumar \cite{Gupta98Critical}, is based on an asymptotic
analysis of large-scale random networks, which considers a
network of $n$ nodes that are  \emph{i.i.d.} on an area with an underlying uniform distribution. A
pair of nodes are directly connected iff their Euclidean
distance is smaller than or equal to a given threshold $r\left(n\right)$,
independent of other connections. Some interesting results are obtained
on the value of $r\left(n\right)$ required for the above network
to be \emph{asymptotically almost surely} connected as $n\rightarrow\infty$.
In \cite{Mao11ConnectivityInfinite,Mao11ConnectivityIsolated}, the authors
extended the above results by Penrose and Gupta and Kumar from the
unit disk model to a random connection model, in which any pair of
nodes separated by a displacement $\boldsymbol{x}$ are directly connected
with probability $g\left(\boldsymbol{x}\right)$, independent of other
connections. 
The analytical
techniques used in this approach have some intrinsic connections
to continuum percolation theory \cite{Meester96Continuum} which
is usually based on a network setting with nodes Poissonly distributed
in an infinite area and studies the conditions required for
the network to have a connected component containing an infinite number
of nodes (in other words, the network \emph{percolates}). We refer readers to \cite{Mao10TowardsExtended} for a more comprehensive review of work in the area.

The second approach is based on a deterministic setting and studies
the connectivity and other topological properties of a network using
 algebraic graph theory. Specifically, consider a network with a
set of $n$ nodes. Its property can be studied using its \emph{underlying
graph} $G\left(V,E\right)$, where $V\triangleq\left\{ v_{1},\ldots,v_{n}\right\} $
denotes the vertex set and $E$ denotes the edge set. The underlying
graph is obtained by representing each node in the network uniquely
using a vertex and the converse. An undirected edge exists between
two vertices iff there is a direct connection (or link) between the
associated nodes\footnote{In this paper, we limit our discussions to a \emph{simple graph} (network) where there is at most one edge (link) between a pair of vertices
(nodes) and an undirected graph.}. Define an \emph{adjacency matrix} $A_{G}$ of the graph $G\left(V,E\right)$
to be a symmetric $n\times n$ matrix whose $\left(i,j\right)^{th},i\neq j$
entry is equal to one if there is an edge between $v_{i}$ and $v_{j}$
and is equal to zero otherwise. Further, the diagonal entries
of $A_{G}$ are all equal to zero. The \emph{eigenvalues of the graph
}$G\left(V,E\right)$ are defined to be the eigenvalues of $A_{G}$.
The network connectivity information, e.g. connectivity and
$k$-connectivity, is entirely contained in its adjacency matrix.
Many interesting connectivity and topological properties of the network
can be obtained by investigating the eigenvalues of its underlying
graph. For example, let $\mu_{1}\geq\ldots\geq\mu_{n}$ be
the eigenvalues of a graph $G$. If $\mu_{1}=\mu_{2}$, then
$G$ is disconnected. If $\mu_{1}=-\mu_{n}$ and $G$ is not
empty, then at least one connected component of $G$ is nonempty
and bipartite \cite[p. 28-6]{Hogben07Handbook}. If the number of
distinct eigenvalues of $G$ is $r$, then $G$ has a \emph{diameter} of at most $r-1$ \cite{Biggs74Algebraic}. Some researchers have also
studied the properties of the underlying graph using its Laplacian matrix \cite{Mohar92Laplace},
where the Laplacian matrix of a graph $G$ is defined as $L_{G}\triangleq D-A_{G}$
and $D$ is a diagonal matrix with degrees of vertices in $G$ on the diagonal.
Particularly, the \emph{algebraic connectivity} of a graph $G$ is
the second-smallest eigenvalue of $L_{G}$ and it is greater
than $0$ iff $G$ is a connected graph.  We refer readers to \cite{Biggs74Algebraic} and \cite{Godsil01Algebraic} for a comprehensive treatment of the topic. Reference \cite{Hogben07Handbook}
provides a concise summary of major results in the area.

The most related research to the work to be presented in this paper is possibly the more recent work of
Brooks et al. \cite{Brooks07Mobile}. In \cite{Brooks07Mobile} Brooks
\emph{et al.} considered a probabilistic version of the adjacency
matrix and defined a \emph{probabilistic adjacency
matrix}
as a $n\times n$ square matrix $M$ whose $(i,j)^{th}$ entry $m_{ij}$
represents the probability of having a direct connection between distinct
nodes $i$ and $j$, and $m_{ii}=0$. They established that the probability
that there exists at least one walk of length $z$ between nodes $i$
and $j$ is $m_{ij}^{z}$, where $m_{ij}^{z}$ is the $(i,j)^{th}$
entry of $M\otimes M\otimes\cdots\otimes M$ ($z$ times). Here $C\triangleq A\otimes B$ is defined by
 $C_{ij}=1-\underset{l\neq i,j}{\prod}\left(1-A_{il}B_{lj}\right)$
where $A_{ij}$, $B_{ij}$ and $C_{ij}$ are the $(i,j)^{th}$ entries
of the $n\times n$ square matrix $A$, $B$ and $C$ respectively.
A \emph{walk} of length $z$ between nodes $i$ and $j$ is a sequence
of $z$ edges, where the first
edge starts at $i$, the last edge ends at $j$, and the starting
vertex of each intermediate edge is the ending vertex of its preceding
edge. A \emph{path} of length $z$ between nodes $i$ and $j$ is
a walk of length $z$ 
in which the edges
are distinct.

Despite intensive research in the area, there is an important category
of problems that remain unsolved: how to measure the
{\it quality of connectivity} of a wireless multi-hop network which has
 a realistic number of nodes,  \emph{not necessarily}
large enough to warrant the use of asymptotic analysis, and has unreliable
connections, reflecting the inherent unreliable
characteristics of wireless communications? The quality of connectivity measures how easily and reliably a packet sent by a  node can reach another  node. It complements the use of \emph{capacity} to measure the quality of a network in saturated traffic scenarios and provides a native measure of the quality of (end-to-end) network connections. In the following paragraphs,
we elaborate on the above question using using two examples.

\begin{example}
\label{exa:probabilistic connectivity matrix}
Consider a network with
a fixed number of nodes with known transmission power to be deployed
in a certain environment. Assume that the wireless propagation model in that
environment is known and its characteristics have been quantified through \emph{a priori} measurements or empirical
estimation. Further, a link exists between two nodes iff the received signal strength from one node at the other node,
whose propagation follows the wireless propagation model, is greater
than or equal to a predetermined threshold \emph{and} the same is
also true in the opposite direction. 
One can then find the probability that a link exists between two    nodes
at two fixed locations: It is  determined by the probability that
the received signal strength is greater than or equal to the pre-determined
threshold. Two related questions can be asked: a) If these nodes
are deployed at a set of known locations, what is the quality of connectivity
of the network, measured by the probability that there is a path between
any two nodes, as compared to node deployment at another set of locations?
b) How to optimize the node deployment to maximize the quality of
connectivity?
\end{example}

\begin{example}
\label{exa:transmission quality}Consider a network with a fixed number
of nodes. The transmission between a pair of nodes with a direct connection,
say $v_{i}$ and $v_{j}$, may fail with a known probability, say
$a_{ij}$, quantifying  the inherent unreliable characteristics
of wireless communications. There are no  direct connections between
some pairs of nodes because the probability of   successful transmission between them
is too low to be acceptable. How to measure the quality of connectivity
of such a network, in the sense that a packet transmitted from one
node can easily and reliably reach another node via a multi-hop path. Will a single
``good'' path between a pair of nodes be more preferable than
multiple ``bad'' paths? These are further illustrated using Fig.
\ref{fig:quality of connectivity example 1} and \ref{fig:quality of connectivity example 2}.
\end{example}
\begin{figure}
\begin{centering}
\includegraphics[scale=0.45]{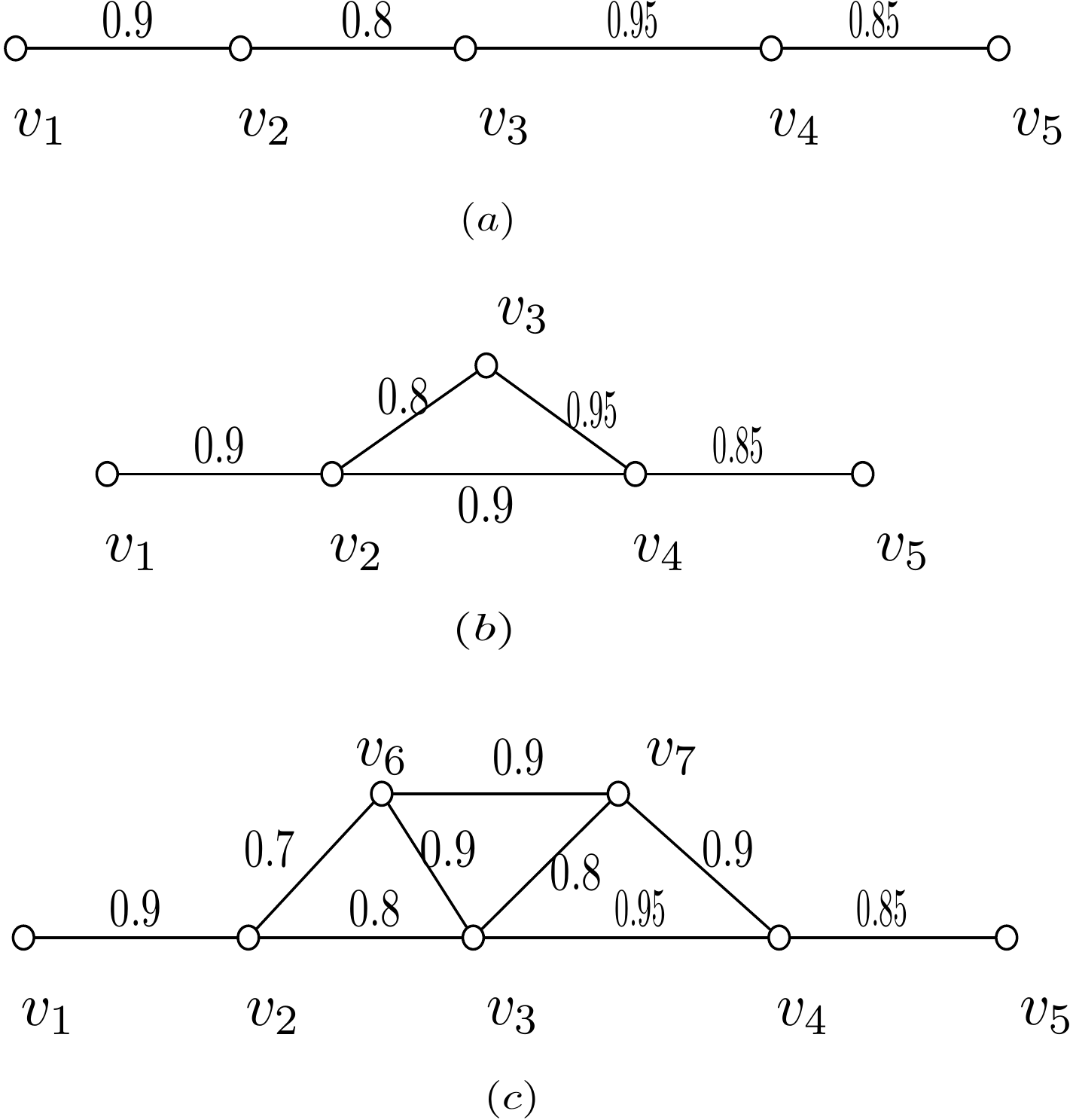}
\par\end{centering}
\caption{An illustration of networks with different quality of connectivity.
A solid line represents a direct connection between two nodes and
the number beside the line represents the corresponding transmission
successful probability. The networks shown in (a), (b), and (c) are
all connected networks but not 2-connected networks, i.e. their connectivity
cannot be differentiated using the k-connectivity concept. However
intuitively the quality of the network in (b) is better
than that of the network in (a) because of the availability of the
additional high-quality link between $v_{2}$ and $v_{4}$ in
(b). The quality of the network in (c) is
even better because of the availability of the additional nodes and
the associated high-quality links, hence additional routes, if these additional nodes act as relay nodes only. If these additional nodes also generate their own traffic, it is uncertain whether the quality of the network in (c) is better or not.
Therefore it is important to develop a measure to quantitatively compare
the quality of connectivity (for the networks in (a) and (b)) and
to evaluate the benefit of additional nodes on connectivity (for the
network in (c)). \label{fig:quality of connectivity example 1}}
\end{figure}
\begin{figure}
\begin{centering}
\includegraphics[scale=0.45]{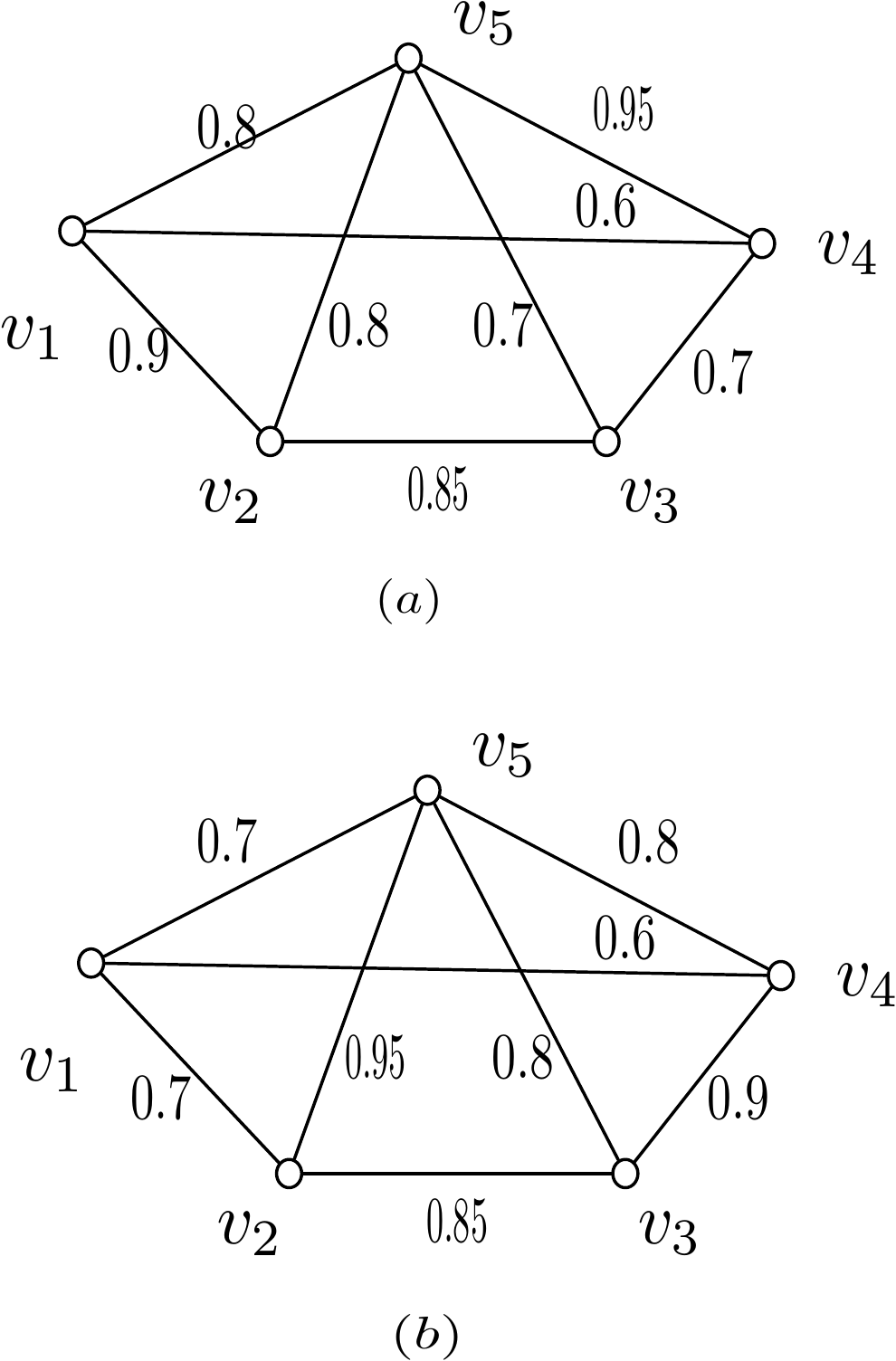}
\par\end{centering}
\caption{The networks shown in (a) and (b) have the same topology but different
link quality. It is difficult to compare the quality of the two networks. \label{fig:quality of connectivity example 2}}
\end{figure}

In this paper, we explore the use of probabilistic connectivity matrix,
a concept to be defined later in Section \ref{sec:probabilistic connectivity matrix},
as a possible tool to measure the quality of network connectivity.
Some interesting properties of the probabilistic connectivity matrix
and their connections to the quality of connectivity are demonstrated. 

The rest of the paper is organized as follows. Section \ref{sec:probabilistic connectivity matrix}
defines the network settings, the probabilistic connectivity matrix
and gives a method to compute the matrix. Section \ref{sec:General-Properties-of probabilistic connectivity matrix}
introduces certain inequalities associated with the entries  of the probabilistic connectivity matrix. Section \ref{sec:qprop} proves several important results about the   probabilistic connectivity matrix. These directly associate the largest eigenvalue of the  probabilistic connectivity matrix to the quality of connectivity  and
expose a structure that holds the promise of facilitating associated optimization tasks. 
 Section \ref{sec:Conclusions}
concludes the paper and discusses future work.

\section{Definition and Construction of the Probabilistic Connectivity Matrix\label{sec:probabilistic connectivity matrix}}

In this section we define the network to be studied, its probabilistic
adjacency matrix and probabilistic connectivity matrix, and gives
an approach to computing the probabilistic connectivity matrix.

Consider a network of $n$ nodes. For some pair of nodes, an edge
(or link) may exist with a non-negligible probability. The edges
are considered to be undirected. That is, if a node $v_{i}$ is
connected to a node $v_{j}$, then the node $v_{j}$ is also connected
to the node $v_{i}$. Further, it is assumed that the event that there
is an edge between a pair of nodes and the event that there is an
edge between another distinct pair of nodes are independent.

Denote the underlying graph of the above network by $G\left(V,E\right)$,
where $V=\left\{ v_{1},\ldots,v_{n}\right\} $ is the vertex set and
$E=\left\{ e_{1},\ldots,e_{m}\right\} $ is the edge set, which contains
the set of \emph{all possible} edges. Here both the vertices and the
edges are indexed from $1$ to $n$ and from $1$ to $m$ respectively.
For convenience, in some parts of this paper we also use the symbol
$e_{ij}$ to denote an edge between vertices $v_{i}$ and $v_{j}$
when there is no confusion. We associate with each edge $e_{i}$,
$i\in\left\{ 1,\ldots m\right\} $, an indicator random variable $I_{i}$
such that $I_{i}=1$ if the edge $e_{i}$ exists; $I_{i}=0$ if the
edge $e_{i}$ does not exist. The indicator random variables $I_{ij}$,
$i\neq j$ and $i,j\in\left\{ 1,\ldots n\right\} $, are defined analogously. 

In the following, we give a definition of the probabilistic adjacency
matrix:
\begin{definitn}
The probabilistic adjacency matrix of $G\left(V,E\right)$, denoted
by $A_{G}$, is a $n\times n$ matrix such that its $(i,j)^{th}$,
$i\neq j$, entry $a_{ij}\triangleq\Pr\left(I_{ij}=1\right)$ and
its diagonal entries are all equal to $1$.
\end{definitn}

Due to the undirected property of an edge mentioned above, 
 $A_{G}$ is a symmetric matrix, i.e.  $a_{ij}=a_{ji}$.
 Note that the diagonal entries of $A_{G}$ are defined to be
$1$, which is different from that common in the literature. In \cite{Mao09Graph}
we have discussed the implication of this definition in the context
of mobile ad-hoc networks. This treatment of the diagonal entries
can be associated with the fact that a node in the network can keep
a packet until better transmission opportunity arises when it finds
the wireless channel busy. 

The probabilistic connectivity matrix is defined in the following:
\begin{definitn}
The probabilistic connectivity matrix of $G\left(V,E\right)$, denoted
by $Q_{G}$, is a $n\times n$ matrix such that its $(i,j)^{th}$,
$i\neq j$, entry is the probability that there exists a path between
vertices $v_{i}$ and $v_{j}$, and its diagonal entries are
all equal to $1$.
\end{definitn}

As a ready consequence of the symmetry of $A_{G}$, $Q_{G}$ is also
a symmetric matrix. 

Given the probabilistic adjacency matrix $A_{G}$, the probabilistic
connectivity matrix $Q_{G}$ is fully determined. However the computation
of $Q_{G}$ is not trivial because for a pair of vertices $v_{i}$
and $v_{j}$, there may be multiple paths between them and some of
them may share common edges, i.e. are not \emph{independent}. In the
following paragraph, we give an approach to computing the probabilistic
connectivity matrix.

Let $\left(I_{1},\ldots,I_{m}\right)$ be a particular instance of
the indicator random variables associated with an instance of the
random edge set. Let $\left.Q_{G}\right|\left(I_{1},\ldots,I_{m}\right)$
be the connectivity matrix of $G$ \emph{conditioned on }$\left(I_{1},\ldots,I_{m}\right)$.
The $\left(i,j\right)^{th}$ entry of $\left.Q_{G}\right|\left(I_{1},\ldots,I_{m}\right)$
is either $0$, when there is no path between $v_{i}$ and $v_{j}$,
or $1$ when there exists a path between $v_{i}$ and $v_{j}$ (see also Lemma \ref{lem:q_ij must be either 0 or 1} in Appendix I). The
diagonal entries of $\left.Q_{G}\right|\left(I_{1},\ldots,I_{m}\right)$
are always $1$. Conditioned on $\left(I_{1},\ldots,I_{m}\right)$,
$G\left(V,E\right)$ is just a deterministic graph. Therefore the
entries of $\left.Q_{G}\right|\left(I_{1},\ldots,I_{m}\right)$ can
be efficiently computed using a search algorithm, such as breadth-first
search. Given $\left.Q_{G}\right|\left(I_{1},\ldots,I_{m}\right)$,
$Q_{G}$ can be computed using the following equation:
\begin{equation}
Q_{G}=E\left(\left.Q_{G}\right|\left(I_{1},\ldots,I_{m}\right)\right)\label{eq:Computation of the probabilistic connectivity matrix}\end{equation}
where the expectation is taken over all possible instances of $\left(I_{1},\ldots,I_{m}\right)$. 

The approach suggested in the last paragraph is essentially a brute-force
approach to computing $Q_{G}$. We expect that more efficient algorithms
can be designed to compute $Q_{G}$. However the main focus of the
paper is on exploring the properties of $Q_{G}$ that facilitate the
connectivity analysis and an extensive discussion of the algorithms
to compute $Q_{G}$ is beyond the scope of the paper.

\begin{remrk}
For simplicity, the terms used in our discussion are based on the
problems  in Example \ref{exa:probabilistic connectivity matrix}.
The discussion however can be easily adapted to the analysis of the
problems in Example \ref{fig:quality of connectivity example 2}.
For example, if $a_{ij}$ is defined to be the probability that a
transmission between nodes $v_{i}$ and $v_{j}$ is successful, the
$\left(i,j\right)^{th}$ entry of the probabilistic connectivity matrix
$Q_{G}$ computed using \eqref{eq:Computation of the probabilistic connectivity matrix}
then gives the probability that a transmission from $v_{i}$ to $v_{j}$
via a multi-hop path is successful under the best routing algorithm,
which can always find a shortest and error-free path between from
$v_{i}$ to $v_{j}$ if it exists, or alternatively, the probability
that a packet broadcast from $v_{i}$ can reach $v_{j}$ where each
node receiving the packet only broadcasts the packet once. Therefore
the $\left(i,j\right)^{th}$ entry of $Q_{G}$ can be used as a quality
measure of the end-to-end paths between $v_{i}$ and $v_{j}$, which
takes into account the fact that availability of an extra path between
a pair of nodes can be exploited to improve the probability of successful
transmissions.
\end{remrk}

\section{Some Key Inequalities for Connection Probabilities\label{sec:General-Properties-of probabilistic connectivity matrix}}

The entries of the probabilistic connectivity matrix give an intuitive
idea about the overall quality of end-to-end paths in a network. In this section, we provide some important inequalities  that may facilitate
the analysis of the quality of connectivity. Some of these inequalities are exploited in the next section 
to establish some key properties of the probabilistic connection matrix itself.

We first introduce some concepts and results that are required for
the further analysis of the probabilistic connectivity matrix $Q_{G}$.

For a random graph with a given set of vertices, a particular event
is \emph{increasing} if the event is preserved when more edges are
added into the graph. An event is \emph{decreasing} if its complement
is increasing. 

The following theorems on FKG inequality and BK inequality respectively
are used:
\begin{thm}
\cite[Theorem 1.4]{Meester96Continuum} \label{thm:FKG inequality}
(FKG Inequality) If events $A$ and $B$ are both increasing events
or decreasing events depending on the state of finitely many edges,
then\[
\Pr\left(A\cap B\right)\geq\Pr\left(A\right)\Pr\left(B\right)\]
\end{thm}

\begin{thm}
\cite[Theorem 1.5]{Berg85Inequalities,Meester96Continuum} \label{thm:BK inequality for increasing events}
(BK Inequality) If events $A$ and $B$ are both increasing events
depending on the state of finitely many edges, then\[
\Pr\left(A\Box B\right)\leq\Pr\left(A\right)\Pr\left(B\right)\]
where for two events $A$ and $B$, $A\Box B$ denotes the event that
there exist two \emph{disjoint} sets of edges such that the first
set of edges guarantees the occurrence of $A$ and the second set
of edges guarantees the occurrence of $B$.
\end{thm}
There is a recent extension of Theorem \ref{thm:BK inequality for increasing events}
to two arbitrary events, i.e. events $A$ and $B$ in Theorem \ref{thm:BK inequality for increasing events}
do not have to be increasing events \cite{Borgs98The}.

Denote by $\xi_{ij}$ the event that there is a path between vertices
$v_{i}$ and $v_{j}$, $i\neq j$. Denote by $\xi_{ikj}$ the event
that there is a path between vertices $v_{i}$ and $v_{j}$ \emph{and}
that path passes through the third vertex $v_{k}$, where $k\in\Gamma_{n}\backslash\left\{ i,j\right\} $
and $\Gamma_{n}$ is the set of indices of all vertices. Denote by
$\eta_{ij}$ the event that there is an edge between vertices
$v_{i}$ and $v_{j}$. Denote by $\pi_{ikj}$ the event that there is a path
between vertices $v_{i}$ and $v_{k}$ \emph{and} there is a path
between vertices $v_{k}$ and $v_{j}$, where $k\in\Gamma_{n}\backslash\left\{ i,j\right\} $.
Obviously \begin{equation}
\pi_{ikj}\Rightarrow\xi_{ij}\label{eq:relation between joint paths and individual paths}\end{equation}
It can also be shown from the above definitions that
\begin{equation}
\xi_{ij}=\eta_{ij}\cup\left(\cup_{k\neq i,j}\xi_{ikj}\right)\label{eq:equation on the event of a path i and j}\end{equation}

Let $q_{ij}$, $i\neq j$, be the $\left(i,j\right)^{th}$ entry of
$Q_{G}$, i.e. $q_{ij}=\Pr\left(\xi_{ij}\right)$. The following lemma
can be readily obtained from the FKG inequality and the above definitions
\begin{lemma}
\label{lem:an inequality on entries of Q - lower bound}For two distinct
indices $i,j\in\Gamma_{n}$ and $\forall k\in\Gamma_{n}\backslash\left\{ i,j\right\} $
\begin{equation}
q_{ij}\geq\max_{k \in \Gamma_{n}\backslash\left\{ i,j\right\} }q_{ik}q_{kj}\label{eq:a lower bound on q_ij}\end{equation}
\end{lemma}
\begin{proof}
It follows readily from the above definitions that the event $\xi_{ij}$
is an increasing event. Due to \eqref{eq:relation between joint paths and individual paths}
and the FKG inequality:
\begin{equation}
\Pr\left(\xi_{ij}\right) \geq \Pr\left(\pi_{ikj}\right)=  \Pr\left(\xi_{ik}\cap\xi_{kj}\right)
\geq \Pr\left(\xi_{ik}\right)\Pr\left(\xi_{kj}\right)\label{eq:an application of FKG}
\end{equation}
The conclusion follows.
\end{proof}
Lemma \ref{lem:an inequality on entries of Q - lower bound} gives
a lower bound of $q_{ij}$. The following lemma gives an upper bound
of $q_{ij}$:
\begin{lemma}
\label{lem:an inequality on entries of Q - upper bound}For two distinct
indices $i,j\in\Gamma_{n}$ and $\forall k\in\Gamma_{n}\backslash\left\{ i,j\right\} $, 
\begin{equation}
q_{ij}\leq1-\left(1-a_{ij}\right)\prod_{k\in\Gamma_{n}\backslash\left\{ i,j\right\} }\left(1-q_{ik}q_{kj}\right)\label{eq:an upper bound on q_ij}\end{equation}
where $a_{ij}=\Pr\left(\eta_{ij}\right)$.\end{lemma}
\begin{proof}
We will first show that $\xi_{ikj}\Leftrightarrow\xi_{ik}\Box\xi_{kj}$.
That is, the occurrence of the event $\xi_{ikj}$ is a sufficient
and necessary condition for the occurrence of the event $\xi_{ik}\Box\xi_{kj}$.

Using the definition of $\xi_{ikj}$, occurrence of $\xi_{ikj}$ means
that there is a path between vertices $v_{i}$ and $v_{j}$ \emph{and}
that path passes through vertex $v_{k}$. It follows that there exist
a path between vertex $i$ and vertex $v_{k}$ \emph{and }a path between
vertex $v_{k}$ and vertex $v_{j}$ \emph{and} the two paths do not
have edge(s) in common. Otherwise, 
it will contradict the definition of $\xi_{ikj}$,
particularly as the definition of a path requires the edges to
be distinct. Therefore $\xi_{ikj}\Rightarrow\xi_{ik}\Box\xi_{kj}$. Likewise, 
$\xi_{ikj}\Leftarrow\xi_{ik}\Box\xi_{kj}$ also follows directly from
the definitions of $\xi_{ikj}$, $\xi_{ik}$, $\xi_{kj}$ and $\xi_{ik}\Box\xi_{kj}$.
Consequently 
\begin{equation}
\Pr\left(\xi_{ikj}\right)=\Pr\left(\xi_{ik}\Box\xi_{kj}\right)\leq\Pr\left(\xi_{ik}\right)\Pr\left(\xi_{kj}\right)\label{eq:an application of the BK inequality}\end{equation}
where the inequality is a direct result of the BK inequality.

With a little bit abuse of the terminology, in the following derivations
we also use $\xi_{ikj}$ to represent the set of edges that make the
event $\xi_{ikj}$ happen, and use $\eta_{ij}$ to denote the edge
between vertices $v_{i}$ and $v_{j}$.

Note that the set of edges $\cup_{k\in\Gamma_{n}\backslash\left\{ i,j\right\} }\xi_{ikj}$
does not contain $\eta_{ij}$. Therefore using \eqref{eq:equation on the event of a path i and j}
and independence of edges (used in the third step)
\begin{eqnarray}
q_{ij} & = & \Pr\left(\eta_{ij}\cup\left(\cup_{k\in\Gamma_{n}\backslash\left\{ i,j\right\} }\xi_{ikj}\right)\right)\nonumber \\
 & = & 1-\Pr\left(\overline{\eta_{ij}}\cap\left(\overline{\cup_{k\in\Gamma_{n}\backslash\left\{ i,j\right\} }\xi_{ikj}}\right)\right)\nonumber \\
 & = & 1-\left(1-a_{ij}\right)\Pr\left(\cap_{k\in\Gamma_{n}\backslash\left\{ i,j\right\} }\overline{\xi_{ikj}}\right)\nonumber \\
 & \leq & 1-\left(1-a_{ij}\right)\prod_{k\in\Gamma_{n}\backslash\left\{ i,j\right\} }\Pr\left(\overline{\xi_{ikj}}\right)\label{eq:use of FKG inequality}\\
 & = & 1-\left(1-a_{ij}\right)\prod_{k\in\Gamma_{n}\backslash\left\{ i,j\right\} }\left(1-\Pr\left(\xi_{ikj}\right)\right)\nonumber \\
 & \leq & 1-\left(1-a_{ij}\right)\prod_{k\in\Gamma_{n}\backslash\left\{ i,j\right\} }\left(1-q_{ik}q_{kj}\right)\label{eq:use of BK inequality}\end{eqnarray}
where in \eqref{eq:use of FKG inequality}, FKG inequality and the obvious fact that $\overline{\xi_{ikj}}$ is a decreasing event are used and the last step results due to
\eqref{eq:an application of the BK inequality}.
\end{proof}
When there is no edge between vertices $v_{i}$ and $v_{j}$, which
is the generic case, the upper and lower bounds in Lemmas \ref{lem:an inequality on entries of Q - lower bound}
and \ref{lem:an inequality on entries of Q - upper bound} reduce
to\begin{equation}
\max_{k\in\Gamma_{n}\backslash\left\{ i,j\right\} }q_{ik}q_{kj}\leq q_{ij}\leq1-\prod_{k\in\Gamma_{n}\backslash\left\{ i,j\right\} }\left(1-q_{ik}q_{kj}\right)\label{eq:a combined upper and lower bound on q_ij}\end{equation}

The above inequality sheds insight on how the quality of paths between
a pair of vertices is related to the quality of paths between other
pairs of vertices. It can be possibly used to determine the most effective
way of improving the quality of a particular set of paths by improving
the quality of a particular (set of) edge(s), or equivalently what can be
reasonably expected from an improvement of a particular edge
on the quality of end-to-end paths.

The following lemma further shows that relation among entries of the
path matrix $Q_{G}$ can be further used to derive some topological
information of the graph.
\begin{lemma}
\label{lem:critical node}If $q_{ij}=q_{ik}q_{kj}$ for three distinct
vertices $v_{i}$, $v_{j}$ and $v_{k}$, the vertex set $V$ of the
underlying graph $G\left(V,E\right)$ can be divided into three non-empty
and non-intersecting sub-sets $V_{1}$, $V_{2}$ and $V_{3}$ such
that $v_{i}\in V_{1}$, $v_{j}\in V_{3}$ and $V_{2}=\left\{ v_{k}\right\} $
\emph{and} any possible path between a vertex in $V_{1}$ and a vertex
in $V_{2}$ must pass through $v_{k}$. Further, for any pair of vertices
$v_{l}$ and $v_{m}$, where $v_{l}\in V_{1}$ and $v_{m}\in V_{3}$,
$q_{lm}=q_{lk}q_{km}$.\end{lemma}
\begin{proof}
Using \eqref{eq:an application of FKG} in the second step, it follows  that\begin{eqnarray*}
q_{ij} & = & \Pr\left(\left(\xi_{ij}\backslash\pi_{ikj}\right)\cup\pi_{ikj}\right)\\
 & = & \Pr\left(\xi_{ij}\backslash\pi_{ikj}\right)+\Pr\left(\pi_{ikj}\right)\\
 & \geq & \Pr\left(\xi_{ij}\backslash\xi_{ikj}\right)+q_{ik}q_{kj}\end{eqnarray*}
Therefore $q_{ij}=q_{ik}q_{kj}$ implies that $\Pr\left(\xi_{ij}\backslash\pi_{ikj}\right)=0$
or equivalently $\xi_{ij}\Leftrightarrow\pi_{ikj}$

Further, $\Pr\left(\xi_{ij}\backslash\pi_{ikj}\right)=0$
implies that a \emph{possible} path (i.e. a path with a non-zero probability)
connecting $v_{i}$ and $v_{k}$ and a \emph{possible} path connecting
$v_{k}$ and $v_{j}$ cannot have any edge in common. Otherwise a
path from $v_{i}$ to $v_{j}$, bypassing $v_{k}$, exists with a
non-zero probability which implies $\Pr\left(\xi_{ij}\backslash\xi_{ikj}\right)>0$.
The conclusion follows readily that if $q_{ij}=q_{ik}q_{kj}$
for three distinct vertices $v_{i}$, $v_{j}$ and $v_{k}$, the vertex
set $V$ of the underlying graph $G\left(V,E\right)$ can be divided
into three non-empty and non-overlapping sub-sets $V_{1}$, $V_{2}$
and $V_{3}$ such that $v_{i}\in V_{1}$, $v_{j}\in V_{3}$ and $V_{2}=\left\{ v_{k}\right\} $
\emph{and} a path between a vertex in $V_{1}$ and a vertex in $V_{2}$,
if exists, must pass through $v_{k}$.

Further, for any pair of vertices $v_{l}$ and $v_{m}$, where $v_{l}\in V_{1}$
and $v_{m}\in V_{3}$, it is easily shown that $\Pr\left(\xi_{lm}\backslash\pi_{lkm}\right)=0$.
Due to independence of edges and further using the fact that $\Pr\left(\xi_{lm}\backslash\pi_{lkm}\right)=0$,
it can be shown that 
\begin{eqnarray}
\Pr\left(\xi_{lm}\right) =  \Pr\left(\pi_{lkm}\right) & = & \Pr\left(\xi_{lk}\cap\xi_{km}\right)\nonumber \\
 & = & \Pr\left(\xi_{lk}\right)\Pr\left(\xi_{km}\right)\label{eq:eq:implication of the exclusion - intermediate step}
\end{eqnarray}
where \eqref{eq:eq:implication of the exclusion - intermediate step}
results due to the fact that under the condition of $\Pr\left(\xi_{lm}\backslash\pi_{lkm}\right)=0$,
a path between vertices $v_{l}$ and $v_{k}$ and a path between vertices
$v_{k}$ and $v_{m}$ cannot possibly have any edge in common. 
\end{proof}
An implication of Lemma \ref{lem:critical node} is that for any three
distinct vertices, $v_{i}$, $v_{j}$ and $v_{k}$, if a relationship
$q_{ij}=q_{ik}q_{kj}$ holds, vertex $v_{k}$ must be a \emph{critical}
vertex whose removal will render the graph disconnected.

\section{Properties of the Connectivity Matrix} \la{sec:qprop}

Having established some inequalities obeyed by the entries
of the probabilistic connectivity matrix $Q_{G}$, we now turn to establishing a measure of the quality of network connectivity.
At the core of the development in this section is the following result.

\begin{lemma}
\label{lem:multiaffine function}Each off-diagonal entry of the probabilistic
connectivity matrix $Q_{G}$ is a multiaffine\footnote{A multiaffine function is affine in each variable when the other variables are fixed.} function of $a_{ij}$,
$i\in\left\{ 1,\ldots,n\right\} ,j>i$.\end{lemma}
\begin{proof}
Observe that $a_{ij}=\Pr\left(\eta_{ij}\right)$ and the events $\eta_{ij}$,
$i\in\left\{ 1,\ldots,n\right\} ,j>i$ are independent. The  conclusion
in the lemma follows readily from the fact that the event associated
with each $q_{ij}$, i.e. there exists a path between vertices $v_{i}$
and $v_{j}$, is a union of intersections of these events $\eta_{ij}$,
$i\in\left\{ 1,\ldots,n\right\} ,j>i$. 
\end{proof}

Not only does the multiaffine structure facilitate the proof of the main result below, we comment later in Remark \ref{remark on optimization}, on how it is potentially useful for performing some of the optimization tasks inherent in maximizing connectivity.

A very desirable property of $Q_{G}$ is established below.
\begin{thm}
\label{thm:positive semi-definite matrix}The probabilistic connectivity
matrix $Q_{G}$, defined for the vertex set $V=\{v_1,\cdots, v_n\}$,  is a positive semi-definite matrix. Further, $Q_{G}$
is positive semi-definite but not positive definite iff
there exist distinct $i,j\in\{1,\cdots,n\}$, such that $q_{ij}=1$.\end{thm}
\begin{proof}
See Appendix I.
\end{proof}

Let $\lambda_{1}\geq\ldots\geq\lambda_{n}$ be the eigenvalues of
$Q_{G}$. Note that $\lambda_{1}+\cdots+\lambda_{n}=n$. 
 As an easy consequence of Theorem \ref{thm:positive semi-definite matrix},
$n\geq\lambda_{1}\geq1$ and $1\geq\lambda_{n}\geq0$. In the best
case, $Q_{G}$ is a matrix with all entries equal to $1$. Then $\lambda_{1}=n$
and $\lambda_{2}=\cdots=\lambda_{n}=0$. In the worst case, $Q_{G}$
is an identity matrix. Then $\lambda_{1}=\cdots=\lambda_{n}=1$. This suggests that
$\lambda_{1}$, i.e. the largest eigenvalue of $Q_G$, can be used as a measure of quality of network connectivity
and a larger $\lambda_{1}$ indicates a better quality. 

We will make this idea more concrete in the following analysis. We start our discussion with a connected network and then extend to more generic cases. We will call a network {\it  connected} if for all $i,j\in\{1,\cdots,n\}$,  $q_{ij}>0$. Obviously the probabilistic connectivity matrix of a connected network is \emph{irreducible}  \cite[p. 374]{LT}  as all the entries of the matrix are non-zero. As a measure of the quality of network connectivity, if the path probabilities $q_{ij}$ increase, the largest eigenvalue of the probabilistic connectivity matrix should also increase. This is formally stated in the following theorem:

\begin{thm}\la{tinc}
Let $G(V, E)$ and $G'(V, E')$ be the underlying graphs of two connected networks defined on the same vertex set $V$ but with different link probabilities. Let $Q_G$ and $Q_{G'}$ be the probabilistic connectivity matrices of $G$ and $G'$ respectively and let $\lambda_{\max}\left(Q_G\right)$ and $\lambda_{\max}\left(Q_{G'}\right)$ be the largest eigenvalues of $Q_G$ and $Q_{G'}$ respectively. If $Q_G' - Q_{G}$ is a non-zero, non-negative matrix\footnote{A matrix is \emph{non-negative} if all its entries are greater than or equal to $0$.}, then
\begin{equation}\label{Eq: a strict increase of the largest eigenvalue for connected network}
\lambda_{\max}\left(Q_G\right) < \lambda_{\max}\left(Q_{G'}\right)
\end{equation}
\end{thm}
\begin{proof}
See Appendix II
\end{proof}

After having analyzed the situation for connected networks, we now move on to the discussion of disconnected networks and show that the largest eigenvalue of the probabilistic connectivity matrix of a \emph{component}, a concept defined in the next paragraph, provides a good measure of the quality of connectivity of that component. We will analyze two basic situations. Results for more complicated scenarios can be readily obtained from these results and Theorem \ref{tinc}, which applies  to a connected network.

If the network is not  connected, i.e. some entries of its probabilistic connectivity matrix is $0$, it can be easily shown that the network can be decomposed into \emph{disjoint components}. A \emph{component} is a maximal set of vertices where the probability that there is a path between any pair of vertices in the component is greater than zero. Two components are said to be \emph{disjoint} if the probabilities that there is a path between any vertex in the first component and any vertex in the second component are all zeros.

Let the total number of components in the network be $k$ and the number of nodes in the $i^{th}$, $1\leq i \leq k$, component be $n_i$. Further, denote the vertex set of the $i^{th}$ component by $V_i$ and denote the subgraph induced on $V_i$ by $G_i$. Without loss of generality, we assume that the nodes in the network are properly labeled such that
\begin{equation}
V_i = \{v_{\sum_{j=1}^{i-1} n_j+1},  \cdots, v_{\sum_{j=1}^{i-1} n_j+n_i}\}
\end{equation} 
Let $Q_{G_i}$ be the probabilistic connectivity matrix of $G_i$. It follows that the probabilistic connectivity matrix of the network $Q_G$ can be expressed in the form of  $Q_{G_i}$, $1\leq i \leq k$, as
\begin{equation}
Q_{G}=\diag \{Q_{G_1}, \cdots, Q_{G_k}\}
\end{equation}
and
\begin{equation}
\lambda_{\max}\left(Q_{G} \right)=\max \{ \lambda_{\max}\left(Q_{G_1} \right), \cdots, \lambda_{\max}\left(Q_{G_k} \right)\}
\end{equation}

We consider two basic situations: a) there are increases in some entries of $Q_{G}$ from non-zero values but such increases do not change the number of components in the network; b) there are increases in some entries of $Q_{G}$ from zero to non-zero values and such increases reduce the number of components in the network.

Under situation a), the vertex set of each component does not change. Let $Q_{G_i}$ be the probabilistic connectivity matrix of a component whose path probabilities have been increased and let $Q_{G'_i}$ be the probabilistic connectivity matrix of the component after the change in path probabilities. Obviously $Q_{G'_i}-Q_{G_i}$ is a non-zero, non-negative symmetric matrix. It then follows easily from Theorem \ref{tinc} that $\lambda_{\max}\left( Q_{G'_i} \right) > \lambda_{\max}\left( Q_{G_i} \right)$.  Depending on whether $\lambda_{\max}\left( Q_{G'_i} \right)$ is greater than $\lambda_{\max}\left(Q_{G} \right)$ or not however, the largest eigenvalue of the probabilistic connectivity matrix of the network may or may not increase.

We now move on to evaluate situation b) and consider a simplified scenario where increases in the path probabilities merge two originally disjoint components, denoted by $G_i$ and $G_j$. The more complicated scenario where increases in the path probabilities join more than two originally disjoint components can be obtained recursively as an  extension of the above simplified scenario. Let $G'$ be the underlying graph of the network after increases in path probabilities and let $G'_{ij}$ be the subgraph in $G'$ induced on the vertex set $V_i \cup V_j$. Obviously $Q_{G'_{ij}}$ is an irreducible matrix and the following result can be established.
\begin{lemma}\label{lem: proof of eigenvalue of joint component} 
Under the above settings,
\begin{equation}
\lambda_{\max}\left( Q_{G'_{ij}} \right) > \lambda_{\max}\left( \diag \{ Q_{G_i}, Q_{G_j}  \} \right)
\end{equation}
\end{lemma}
 The proof of Lemma \ref{lem: proof of eigenvalue of joint component} is omitted due to space limitations.

Thus indeed  the largest eigenvalues of the probabilistic connection matrices associated with disjoint components measure the quality of the components connection. 

\begin{remrk}\label{remark on optimization}
The fact that the largest eigenvalue of the probabilistic connectivity matrix measures connectivity, suggests the following obvious optimization. Modify  one or more $a_{ij}$
under suitable constraints to maximize the largest eigenvalue of the probabilistic connectivity matrix. Results in \cite{ZD} and  \cite{Dasgupta94Lyapunov} suggest that the multiaffine dependence of the $q_{ij}$  on the $a_{ij}$ {\it together with the fact
that $Q_G$ is positive semi-definite}  promise to facilitate such optimization. 
\end{remrk}

\section{Conclusions and Further Work \label{sec:Conclusions}}

In this paper we have explored the use of the probabilistic connectivity matrix as a possible tool to measure the quality of network connectivity. Some interesting properties of the probabilistic connectivity matrix and their connections to the quality of network connectivity were demonstrated. Particularly, the off-diagonal entries of the probabilistic connectivity matrix provide a measure of the quality of end-to-end connections and we have also provided theoretical analysis supporting the use of the largest eigenvalue of the probabilistic connectivity matrix as a measure of the quality of overall network connectivity. The analysis focused on the comparison of networks with the same number of nodes. For networks with different number of nodes, the largest eigenvalue of the probabilistic connectivity matrix normalized by the number of nodes may be used as the quality metric. 

Inequalities between the entries of the probabilistic connectivity matrix were established. These may provide insights into the correlations between quality of end-to-end connections. Further, the probabilistic connectivity matrix was shown to be a positive semi-definite matrix and its off-diagonal entries are multiaffine functions of link probabilities. These two properties are expected to be very helpful in optimization and robust network design, e.g. determining the link whose quality improvement will result in the maximum gain in network quality, and determining quantitatively  the relative criticality of a link to either a particular end-to-end connection or to the entire network.

The results in the paper rely on two main assumptions: the links are symmetric \emph{and}  independent. We expect that our analysis can be readily extended such that the first assumption on symmetric links can be removed -- in fact the results in Section \ref{sec:General-Properties-of probabilistic connectivity matrix} do not need this assumption. While in the asymmetric case the probabilistic connectivity matrix is no longer guaranteed to be positive semi-definite, we conjecture that the largest eigenvalue retains its significance. Discarding the second assumption  requires more work. However, we are encouraged by the following observation. If we introduce conditional edge probabilities into the mix, then $Q_G$ is still a multiaffine function of the $a_{ij}$ and the conditional probabilities. Thus we still expect all the results in Section \ref{sec:qprop} to hold, though the proof may be non-trivial. In real applications link correlations may arise due to both physical layer correlations and correlations caused by traffic congestion.

Another implicit assumption in the paper is that traffic is uniformly distributed and traffic between every source-destination pair is equally important. If this is not the case, a weighted version of the probabilistic connectivity matrix can be contemplated in which the entries of the matrix are weighted by a measure of the importance of the associated source-destination pairs. It remains to be investigated on whether our results can be extended to a weighted probabilistic connectivity matrix.

\section*{Appendix I: Proof of Theorem \ref{thm:positive semi-definite matrix}}
Let $a_{U}^{\left(n\right)}$ be the vector of $a_{ij}$, $i\neq j$
(remember that $a_{ij}$ is the probability that there is an edge
between $v_{i}$ and $v_{j}$): \[
a_{U}^{\left(n\right)}\triangleq\left[a_{ij}\right]_{i=1,j>i}^{n}\]
Also define the set \[
\Pi_{U}^{n}=\left\{ \left.a_{U}^{\left(n\right)}\right|0\leq a_{ij}\leq1,i\in\left\{ 1,\ldots,n\right\} ,j>i\right\} \]
The corners of the above set are given by:\[
\Pi_{U_{c}}^{n}=\left\{ \left.a_{U}^{\left(n\right)}\right|a_{ij}\in\left\{ 0,1\right\} ,i\in\left\{ 1,\ldots,n\right\} ,j>i\right\} \]
These corners will play an important role in the subsequent development.

We observe that the positive semi-definiteness of a matrix is known
to be a convex property, as defined in \cite{Dasgupta94Lyapunov}.
That is, if $A$ and $B$ are positive semi-definite then so is $\left(1-\alpha\right)A+\alpha B$,
$\forall\alpha\in\left[0,1\right]$. In fact one can say a bit more:
\begin{lemma}
\label{lem:linear combination of positive semi-definite matrices}Consider
$n\times n$ matrices $A>0$ and $B\geq0$. Then: $\left(1-\alpha\right)A+\alpha B>0$,
$\forall\alpha\in\left[0,1\right)$\end{lemma}
\begin{proof}
It is well known that there exists a matrix $H$, such that $H\left(A+A^{T}\right)H=I$
and for some diagonal $\Lambda$, $H\left(B+B^{T}\right)H=\Lambda\geq0$.
Then the result follows by noting that $\left(1-\alpha\right)I+\alpha\Lambda>0$,
$\forall\alpha\in\left[0,1\right)$
\end{proof}

Next we provide a key result that exploits the multiaffine dependence of $Q_G$ on the $a_{ij}$. 
\begin{prop}
\label{pro:positive-semidefinite from boundary to general cases}Suppose
$Q_{G}$ is positive semi-definite for all $a_{U}^{\left(n\right)}\in\Pi_{U_{c}}^{n}$.
Then it is positive semi-definite for all $a_{U}^{\left(n\right)}\in\Pi_{U}^{n}$.\end{prop}
\begin{proof}
The combined use of  Lemma \ref{lem:multiaffine function}  and  Corollary 2.1 of \cite{Dasgupta94Lyapunov} proves the result.
\end{proof}

We must next  show that $Q_{G}$ is positive semi-definite for all
$a_{U}^{\left(n\right)}\in\Pi_{U_{c}}^{n}$. The following lemma is
used in the proof of the conclusion that $Q_{G}$ is positive semi-definite
for all $a_{U}^{\left(n\right)}\in\Pi_{U_{c}}^{n}$:
\begin{lemma}
\label{lem:q_ij must be either 0 or 1}Suppose $a_{U}^{\left(n\right)}\in\Pi_{U_{c}}^{n}$,
then for all $i$, $j$, $q_{ij}\in\left\{ 0,1\right\} $.\end{lemma}
\begin{proof}
When $a_{U}^{\left(n\right)}\in\Pi_{U_{c}}^{n}$, either there is
an edge between vertices $v_{i}$ and $v_{j}$ surely when $a_{ij}=1$;
or there is no edge between vertices $v_{i}$ and $v_{j}$ surely
when $a_{ij}=0$. The graph $G\left(V,E\right)$ becomes a deterministic
graph. It follows that either there is a path between vertices $v_{i}$
and $v_{j}$ surely or there is no path between vertices $v_{i}$
and $v_{j}$ surely, i.e. for all $i$, $j$, $q_{ij}\in\left\{ 0,1\right\} $.
\end{proof}

It can be further shown that the following lemma  holds:
\begin{lemma}
\label{lem:identical rows and columns}Suppose for some distinct $i$,
$j$, $q_{ij}=1$. Then row $i$ and row $j$ of $Q_{G}$ are identical,
as are columns $i$ and $j$.\end{lemma}
\begin{proof}
Note that $Q_{G}$ is a symmetric matrix. Thus it suffices to show
that the row property holds. One has \begin{equation}
q_{ij}=q_{ji}=q_{ii}=q_{jj}=1 \label{eq: equality of q_ij and q_ji}\end{equation}
Thus the $i^{th}$ and $j^{th}$ entries of the $i^{th}$ and $j^{th}$
rows are identical. Now consider any $k$ distinct from $i$ and $j$.
Using Lemma \ref{lem:an inequality on entries of Q - lower bound}  and \eqref{eq: equality of q_ij and q_ji}:
\[
q_{ik}\geq q_{ij}q_{jk}=q_{jk} \text{   and   } q_{jk}\geq q_{ij}q_{ik}=q_{ik}\]
It follows that: $q_{jk}=q_{ik}$.
\end{proof}

We need one last lemma to complete the proof.
\begin{lemma}
\label{lem:diagonalization of Q_G}Define $u_{m}$ to be a $m$-vector
of all ones. Suppose $a_{U}^{\left(n\right)}\in\Pi_{U_{c}}^{n}$.
Then for some $k$, there exist positive integers $m_{1}, \cdots, 
m_{k}$ whose sum is $n$, and a permutation matrix $P$, such that:
\begin{equation}
Q_{G}=P\diag\left\{ u_{m_1}u_{m_1}^{T},\ldots,u_{m_k}u_{m_k}^{T}\right\} P^{T}\label{eq:diagonalization of Q_G}\end{equation}
\end{lemma}
\begin{proof}
We prove the lemma by induction. 

As an easy consequence of Lemma \ref{lem:q_ij must be either 0 or 1},
the result clearly holds for $n=2$ because when $n=2$, $Q_{G}$
is either equal to an identity matrix or equal to a matrix of all
ones. 

Now suppose that the lemma holds for all $n\leq m$. For convenience,
we use also $Q_{n}$ to denote $Q_{G}$ when $\left|V\right|=n$.

When $n=m+1$, consider $Q_{m+1}$, corresponding to any $a_{U}^{\left(n\right)}\in\Pi_{U_{c}}^{n}$.
Because of Lemma \ref{lem:q_ij must be either 0 or 1}, all entries
of $Q_{m+1}$ are in $\left\{ 0,1\right\} $. If $q_{ij}=0$, $\forall i\neq j$,
then the results hold with $m_{i}=1$ and $k=n$. Now suppose there exists
some distinct $i$ and $j$ for which $q_{ij}=1$. By symmetrically permuting
$Q_{m+1}$, i.e. by relabeling the nodes if necessary, one can without
loss of generality choose $\left\{ i,j\right\} =\left\{ 1,2\right\} $.
Choose $m_{l}$ to equal the number of ones in the
first row of this possibly permuted $Q_{m+1}$. Through a further
symmetric permutation/relabeling if necessary, without loss of generality
one has:\[
q_{1j}=q_{j1}=1,\;\forall j\in\left\{ 1,\ldots,m_{1}\right\} \]
From Lemma \ref{lem:q_ij must be either 0 or 1}:\[
q_{1j}=q_{j1}=0,\;\forall j\in\left\{ m_{1}+1,\ldots,n\right\} \]
From Lemma \ref{lem:identical rows and columns}, there holds:\[
q_{ij}=\begin{cases}
1 & \forall i,j\in\left\{ 1,\ldots,m_{1}\right\} \\
0 & \forall i\in\left\{ 1,\ldots,m_{1}\right\} \;\textrm{and}\; j\in\left\{ m_{1}+1,\ldots,n\right\} \\
0 & \forall j\in\left\{ 1,\ldots,m_{1}\right\} \;\textrm{and}\; i\in\left\{ m_{1}+1,\ldots,n\right\} \end{cases}\]
Thus after relabeling one can express: \[
Q_{G}=\diag{\left\{ u_{m_1}u_{m_1}^{T},Q_{m+1-m_{1}}\right\} }\]
Further, there is no path from the vertex set $\left\{ v_{1},\ldots,v_{m_{1}}\right\} $
to the vertex set $\left\{ v_{m_{1}},\ldots,v_{n}\right\} $ and vice versa. 

Obviously the entries of $Q_{m+1-m_{1}}$ form legitimate path probabilities
with the vertex set $\left\{ v_{m_{1}},\ldots,v_{n}\right\} $. Then
the inductive hypothesis proves the result.
\end{proof}

We are now ready to prove Theorem \ref{thm:positive semi-definite matrix}. 

Observe the matrix in \eqref{eq:diagonalization of Q_G} is positive
semi-definite. Thus from Lemma \ref{lem:diagonalization of Q_G},
$Q_{G}$ is positive semi-definite for all $a_{U}^{\left(n\right)}\in\Pi_{U_{c}}^{n}$.
It then follows from Proposition \ref{pro:positive-semidefinite from boundary to general cases}
that $Q_{G}$ is positive semi-definite for all $a_{U}^{\left(n\right)}\in\Pi_{U}^{n}$.  Therefore the first part of Theorem \ref{thm:positive semi-definite matrix} that $Q_G$ is a positive semi-definite matrix is proved.

Now we proceed to prove the second part of Theorem \ref{thm:positive semi-definite matrix} that $Q_{G}$
is positive semi-definite but not positive definite iff
there exist distinct $i,j\in\{1,\cdots,n\}$, such that $q_{ij}=1$.

Suppose there exists distinct $i$ and $j$ such that $q_{ij}=1$.
Then from Lemma \ref{lem:identical rows and columns}, at least two
rows of $Q_{G}$ are identical. Thus $Q_{G}$ is singular and 
cannot be positive definite.

It remains to show that if for some $a_{U}^{\left(n\right)}\in\Pi_{U}^{n}$,
all $q_{ij}\neq1$ where $i\neq j$, then $Q_{G}$ is positive definite. We prove this
by induction.

The result is clearly true for $n=2$. Suppose it holds for some $n=m$.
Consider $n=m+1$. 

To establish a contradiction suppose there is a
$a_{U}^{\left(n\right)}\in\Pi_{U}^{n}$ for which all $q_{ij}\neq1$ where $i\neq j$,
and yet $Q_{n}$ is not positive definite. This means all $a_{ij}\neq1$ where $i\neq j$.

If for all $i\in\left\{ 1,\ldots,n-1\right\} $, $a_{in}=0$, then
for all $i\in\left\{ 1,\ldots,n-1\right\} $, $q_{in}=0$. Then vertex $v_{n}$ is disconnected from
the vertex set $\left\{ v_{1},\ldots,v_{n-1}\right\} $ and for all
$\left\{ i,j\right\} \subseteq\left\{ 1,\ldots,n-1\right\} $, $q_{ij}$
are valid path probabilities in the subgraph induced on the vertex set
$\left\{ v_{1},\ldots,v_{n-1}\right\} $. Further $Q_{n}=diag\left\{ Q_{n-1},1\right\} $.
By hypothesis $Q_{n-1}$ is positive definite and thus so is $Q_{n}$.
Because of the resulting contradiction with the hypothesis that $Q_n$ is not positive definite, it follows that for at least one $i\in\left\{ 1,\ldots,n-1\right\} $,
$a_{in}\neq0$. 

Through relabeling, without sacrificing generality,
assume that 
$a_{in}\neq0,\forall i\in\left\{ 1,\ldots,l\right\}$ and $a_{in}=0,\forall i\in\left\{ l+1,\ldots,n-1\right\}$.

Define $R$ to be the probabilistic connectivity matrix obtained by keeping all corresponding $a_{ij}$
the same except $a_{ln}$, which is set to zero. Similarly, define $S$ to be the probabilistic connectivity matrix obtained by keeping
all corresponding $a_{ij}$ the same except $a_{ln}$, which is set to $1$. 

Observe,
because of Lemma \ref{lem:multiaffine function}, that $Q_{n}$ is
a convex combination of $R$ and $S$. Further, due to the hypothesis that $Q_{n}$ is not positive definite, both $R$ and $S$ are positive
semi-definite but neither can be positive definite. Otherwise, using Lemma \ref{lem:linear combination of positive semi-definite matrices}, $Q_{n}$ will be positive definite which leads to a contradiction of the hypothesis. In particular $R$ is positive semi-definite,
but not positive definite. Now working with $R$, if $l>1$, then the probabilistic connectivity 
matrix obtained by keeping all corresponding $a_{ij}$ the same except $a_{2n}$,
which is set to zero, is similarly positive semi-definite, but not
positive definite.

Working recursively in this fashion, the probabilistic connectivity matrix obtained by keeping all
$a_{ij}$ the same except $a_{in}=0$, $\forall i\in\left\{ 1,\ldots,n-1\right\} $,
is positive semi-definite, but not positive definite. Thereby a contradiction is established with the conclusion obtained in the earlier paragraph that  for at least one $i\in\left\{ 1,\ldots,n-1\right\} $,
$a_{in}\neq0$. 

\section*{Appendix II: Proof of Theorem \ref{tinc}}

The proof of this Theorem appeals to the celebrated Perron-Frobenius theorem whose basics we recount below  \cite[p. 536]{LT}.

\begin{thm}\la{tpf}
Suppose a matrix $A\in \mR^{n\times n}$ is non-negative and irreducible. Then the largest eigenvalue of $A$ is simple, positive and has a corresponding eigenvector all whose elements are positive. If $A$ is reducible then the largest eigenvector corresponding to its largest eigenvalue can be chosen to be non-negative.
\end{thm}

We also need the following lemma to prove Theorem \ref{tinc}.



\begin{lemma} \la{lab}
Suppose $A=A^T\neq B=B^T$ are  non-negative, irreducible, real matrices, and $B-A$ is a non-zero, non-negative matrix. Then: $\lambda_{\max}(A)< \lambda_{\max}(B )$.
\end{lemma}
\begin{proof}Observe at least one element of $B-A$ is positive. From Theorem \ref{tpf}, $x\in \mR^n$, the eigenvector corresponding to the largest
eigenvalue of $A$ can be chosen to have all elements positive.
Then the result follows from the fact that:
\beas
\lambda_{\max}(A)x^Tx &=&x^TAx \\
&=&x^TBx- x^T(B-A)x\\
&<&x^TBx \\
&\leq &\lambda_{\max}(B)x^Tx
\eeas
as $B-A$ is a non-zero, non-negative matrix.
\end{proof}

Turning to the proof of Theorem \ref{tinc} we note that the result follows directly from Lemma \ref{lab} and the fact $Q_{G'}$ and $Q_G$ satisfy the requirements of $B$ and $A$,
respectively. 


\end{document}